\documentclass[11pt]{article}

\usepackage{amsmath,amsthm}
\usepackage{amssymb}
\usepackage{amsfonts}
\usepackage{textcomp}
\usepackage{amsthm}
\usepackage{nicefrac}
\usepackage{xspace}
\usepackage[english]{babel}
\usepackage[utf8]{inputenc}
\usepackage[T1]{fontenc}\usepackage{hyperref}
\hypersetup{colorlinks=true,linkcolor=black,citecolor=[rgb]{0.35,0.35,0.75}}
\usepackage{mathtools}
\usepackage[usenames]{color}
\usepackage[svgnames]{xcolor}
\usepackage{tikz}
\usepackage{wrapfig}
\usepackage{todonotes}
\usepackage{complexity}
\usepackage[margin=1.1in]{geometry}
\usepackage{mathpazo}

\theoremstyle{plain}% default
\newtheorem{theorem}{Theorem}[section]
\newtheorem{lemma}[theorem]{Lemma}
\newtheorem{definition}[theorem]{Definition}

\newtheorem{fact}{Fact}

\newcommand{\dist}{\Delta}
\newcommand{\RR}{\mathbb{R}}

\newcommand{\acw}{ACW}
\newcommand{\cw}{CW}

\renewcommand{\tilde}{\widetilde}
\renewcommand{\R}{\mathbb{R}}

\newcommand\ie{{i.e.,}\xspace}

\newcommand{\BUF}{$\mathsf{ULT}$\xspace}
\newcommand{\LCA}{\mathsf{LCA}}
\newcommand{\SETH}{$\mathsf{SETH}$\xspace}
\newcommand{\OVH}{$\mathsf{OVH}$\xspace}
\renewcommand{\CH}{$\mathsf{CH}$\xspace}
\renewcommand{\CP}{$\mathsf{CP}$\xspace}

\newcommand{\B}{$\mathcal{B}_d$\xspace}
\newcommand{\Dun}{$\mathcal{D}_{\texttt{uni}}$\xspace}
\newcommand{\Dpl}{$\mathcal{D}_{\texttt{plant}}$\xspace}

\newif\ifdraft\drafttrue
\draftfalse
\ifdraft

\newcommand\vc[1]{\todo[inline,size=\scriptsize,backgroundcolor=SpringGreen]{#1 - \textbf{Vincent}}}
\newcommand\gl[1]{\todo[inline,size=\scriptsize,backgroundcolor=Pink]{#1 - \textbf{Guillaume}}}
\newcommand\cs[1]{\todo[inline,size=\scriptsize,backgroundcolor=Yellow]{#1 - \textbf{Karthik}}}

\else

\newcommand\vc[1]{}
\newcommand\gl[1]{}
\newcommand\cs[1]{}

\fi

\newcommand{\opt}{_{\mathsf{OPT}}}

\usepackage{setspace}
\title{On Efficient Low Distortion Ultrametric Embedding}

\author{Vincent Cohen-Addad\\ Google\\ \texttt{vcohenad@gmail.com}\and Karthik C.\ S.\\
Tel Aviv University\\
\texttt{karthik0112358@gmail.com}\vspace{0.5cm}\and Guillaume Lagarde\\ LaBRI Bordeaux\\ \texttt{guillaume.lagarde@gmail.com} }

\date{}

\begin{document}

\maketitle
\begin{abstract}
A classic problem in unsupervised learning and data analysis is to find
   \emph{simpler} and \emph{easy-to-visualize} representations of the data
that preserve its essential properties. 
   A widely-used method to preserve the underlying hierarchical structure
of the data while reducing
   its complexity is to find an embedding of the data into a tree or an
ultrametric. The most popular
   algorithms for this task are the classic \emph{linkage} algorithms
(single, average, or complete).
   However, these methods on a data set of $n$ points in $\Omega(\log n)$
   dimensions exhibit a quite prohibitive running time of
$\Theta(n^2)$.

   In this paper, we provide a new algorithm which takes as input a set of
   points $P$ in $\R^d$, and for every $c\ge 1$,  runs in time $n^{1+\frac{\rho}{c^2}}$ (for some universal constant $\rho>1$)
to output an ultrametric $\Delta$ such that for any two points $u,v$ in $P$,
we have $\Delta(u,v)$ is within a multiplicative factor of $5c$ to
   the distance between $u$ and $v$ in the ``best'' ultrametric
representation of $P$. Here, the best ultrametric is the ultrametric $\tilde\Delta$ that minimizes
the maximum distance distortion with respect to the $\ell_2$ distance,
namely that minimizes $\underset{u,v \in P}{\max}\ \nicefrac{\tilde\Delta(u,v)}{\|u-v\|_2}$.

   We complement the above result by showing that under popular complexity theoretic assumptions, for every constant
   $\varepsilon>0$, no algorithm with running
time $n^{2-\varepsilon}$ can distinguish between inputs in $\ell_\infty$-metric that admit isometric embedding and those that   incur a distortion of $\nicefrac{3}{2}$.
   
Finally, we present empirical evaluation on classic machine learning
datasets and show that the output
   of our algorithm is comparable to the output of the linkage
algorithms while achieving a much faster running time.
\end{abstract}
\clearpage

\section{Introduction}

The \emph{curse of dimensionality} has ruthlessly been haunting machine learning and data mining researchers.
On the one hand, high dimensional representation of data elements allows fine-grained description of each datum and 
can lead to more accurate models, prediction and understanding.
On the other hand, obtaining a significant signal in each dimension often requires a huge amount of data and
high-dimensional data requires algorithms that can efficiently handle it.
Hence, computing a \emph{simple} representation of a high-dimensional dataset while preserving its most important properties
has been a central problem
in a large number of communities since the 1950s.

Of course, computing a simple representation of an arbitrary high-dimensional set of data elements
necessarily incurs some information loss. Thus, the main question has been to
find dimensionality reduction techniques that would preserve -- or better, \emph{reveal} -- some structure of the data.
An example of such a successful approach has been the \emph{principal component analysis} which can be used to denoise
a dataset and obtain a low-dimensional representation where `similar' data elements are mapped to close-by locations.
This approach has thus become a widely-used, powerful tool to identify cluster structures in high-dimensional datasets.

Yet, in many cases more complex structures underlie the datasets and
it is crucial to identify this structure.  For example, given similarity
relations between species, computing a phylogenetic tree requires more
than identifying a `flat' clustering structure, it is critical to
identify the whole hierarchy of species.  Thus, computing a simple
representation of an input containing a hierarchical structure has
drawn a lot of attention over the years, in particular from
the computational biology community. The most popular approaches
are arguably the \emph{linkage} algorithms, average-linkage,
single-linkage, Ward's method, and complete-linkage, which produce an
embedding of the original metric into an ultrametric\footnote{An
  ultrametric $(X,\Delta)$ is a metric space where for each $x,y,z \in X$,
  $\Delta(x,y) \le \max(\Delta(x,z),\Delta(z,y))$.},
  see for example the seminal work
of \cite{Carlsson10}. Unfortunately, these approaches come with a
major drawback: all these methods, have quadratic running time\footnote{We would like to note here that the relevant work of \cite{AbboudCH19} only mimics the behavior of average-linkage or ward’s method and does not necessarily output an ultrametric. }
-- even in the best case -- when the input consists of points
in $\Theta(\log n)$ dimensions (where $n$ is the number of points)
making them impractical for most applications nowadays. Obtaining an
efficient algorithm for computing ``good'' hierarchical representation
has thus been a major problem (see Section~\ref{sec:relatedwork} for
more details).

In this paper we are interested in constructing
embeddings that (approximately) preserve the hierarchical structure
underlying the input. For example, given three points
$a,b,c$, we would like that if $a$ is more similar to $b$ than to $c$
(and so $a$ is originally closer to $b$ than to $c$ in the
high-dimensional representation), then the  distance of $a$ to $b$ in the
ultrametric is lower than its distance to $c$. More formally, given a
set of points $X$ in Euclidean space, a \emph{good ultrametric}
representation $\Delta$ is such that for every two points $a,b$ in $X$, we have 
$$\|a-b\|_2 \le \Delta(a,b) \le \alpha\cdot \|a-b\|_2,$$ for the smallest possible
$\alpha$ (see formal definition in Section~\ref{sec:prelims}).
Interestingly, and perhaps surprisingly,
this problem can be solved in $O(n^2d + n^2\log n)$ using an
algorithm by \cite{FKW95}. Unfortunately, this algorithm also suffers
from a quite prohibitive quadratic running time. We thus ask:

 \begin{center}\textit{Is there an
easy-to-implement, efficient algorithm for finding \\
good ultrametric
representation of high-dimensional inputs?}
\end{center}
%% \gl{Not sure but I think the complexity of Farach with, as input, points in euclidean space, is more $O(dn^2 + n^2\log n)$ I think. $dn^2$ to compute distances between all points, and $n^2 \log n$ to sort them (+ compute cut weights etc..). What do you think?}

\subsection{Our Results}
We focus on the problem mentioned above, which we refer to as the  \textsc{Best Ultrametric Fit} problem
(\BUF) and which is formally defined in Section~\ref{sec:prelims}.
We provide a simple algorithm, with running time $O(nd)+ n^{1+O(1/\gamma^2)}$ that returns a $5\gamma$-approximation for
the \BUF problem, or a near-linear time algorithm that returns an $O(\sqrt{\log n})$-approximation.

\vc{I have replace the previous theorem with this one, which is the
correct one I believe}
\begin{theorem}[Upper Bound]
  \label{thm:UB:euclid}
  For any $\gamma > 1$, there is an algorithm that produces a $5\gamma$-approximation
  in time $nd + n^{1+O(1/\gamma^2)}$
  for Euclidean instances of \BUF  of dimension $d$.\\
  Moreover, there is an algorithm that produces an
  $O(\sqrt{\log n})$-approximation
  in time $O(nd + n \log^2 n)$ for Euclidean instances of \BUF  of dimension $d$.
\end{theorem}
%% \begin{theorem}\label{thm:ub}
%%   Given a set $X$ points in $\R^d$, there exists an algorithm with running time
%%   $O(|X|d+ |X|^{1+1/\gamma^2}\log |X|)$ that returns a $5\gamma$-approximation
%%   to the Best Ultrametric Fit under $\ell_{\infty}$ problem.
%% \end{theorem}

From a theoretical point of view, note that we can indeed get rid of the $nd$ dependency in the above theorem and replace it with an optimal bound depending on the number of non-zero coordinates by  applying a sparse Johnson-Lindenstrauss transform in the beginning. Nonetheless, we stuck to the $nd$ dependency as it keeps the presentation of our algorithm simple and clear, and also since this is what we use in the experimental section.

Importantly, and perhaps surprisingly, we show that finding a faster than $n^{2-\varepsilon}$ algorithm for this problem is beyond current techniques.
\begin{theorem}[Lower Bound; Informal version of Theorem~\ref{thm:ellinf}]\label{thm:lb}
  Assuming \SETH, for every $\varepsilon>0$, no algorithm running in time $n^{2-\varepsilon}$
  can determine if an instance of \BUF of points in $\ell_\infty$-metric admits
  an isometric embedding or every embedding has distortion  at least  $3/2$.
\end{theorem}

We also provide inapproximability results for the Euclidean metric by ruling out $(1+o(1))$-approximation algorithms for \BUF running in time $n^{1+o(1)}$ albeit under a more non-standard hypothesis that we motivate and introduce in this paper (see Theorem~\ref{thm:ellp} for details).

\paragraph*{Empirical results}
We implemented our algorithm and performed experiments on three
classic datasets (DIABETES, MICE, PENDIGITS). We compared the results
with classic linkage algorithms (average, complete, single) and Ward's
method from the Scikit-learn library~\cite{scikit-learn}. \vc{add ref
  to scikit learn from their website}\gl{done} For a parameter
$\gamma$ fixed to $\approx 2.5$, our results are as follows. First, as complexity
analysis predicts, the execution of our algorithm is much faster
whenever the dataset becomes large enough: \vc{for which value of
  $\gamma$? to see whether it matches the $n^{1+1/\gamma^2}$
  prediction?} \gl{Parameter set to $2.5$ but this is very rough since it's an asymptotic behavior. Don't understand exactly what do you mean by: ``whether it matches $n^{blabla}$? Do you want to plot the running time curve and see if that looks like a function of the desired form? If yes, we would probably need to plot for lots of datasets, not just 3, no?}  up to $\approx 36$ (resp. $32$, $7$ and $35$) times
faster than average linkage (resp. complete linkage, single linkage
and Ward's method)  for moderate size dataset containing roughly
$10000$ points, and has comparable running time for smaller
inputs. Second, while achieving a much faster running time, the
quality of the ultrametric stays competitive to the distortion
produced by the other linkage algorithms.  Indeed, the maximum
distortion is, on these three datasets, always better than Ward's
method, while staying not so far from the others: in the worst case up to a factor $\approx 5.2$ (resp. $4.3$, $10.5$) against average linkage (resp. complete and single linkages).\vc{maybe states by
  how much worse and see whether it is better than what the theory
  predicts. In which case, we should say so}\gl{done}  This shows that our new
algorithm is a reliable and efficient alternative to the linkage
algorithms when dealing with massive datasets.

%% In conclusion, we highlight several points that
%% could, for future implementations, improve both running times and
%% output qualities.

\subsection{Related Work}
\label{sec:relatedwork}

Strengthening the foundations for hierarchical representation
of complex data has received a lot of attention over the years.
The thorough study of \cite{Carlsson10} has deepened our understanding
of the linkage algorithms and the inputs for which they produce good
representations, we refer the reader to this work for a more complete
introduction to the linkage algorithms.
Hierarchical representation of data and hierarchical clusterings
are similar problems. 
A recent seminal paper by \cite{dasgupta15} phrasing
the problem of computing a good hierarchical clustering as an
optimization problem has sparked
a significant amount of work mixing theoretical and practical results.
\cite{CKMM18, MW17} showed that 
average-linkage achieves a constant factor approximation to (the dual of)
Dasgupta's function and introduced new algorithms with worst-case and
beyond-worst-case guarantees, see also~\cite{RP16,CC17,CKM17,CCN19,CCN18}.
Single-linkage is also known to be helfpul to identify `flat' clusterings
in some specific settings~\cite{Balcan08}.
We would like to point out that this previous work did not consider
the question of producing an ultrametric that is representative of
the underlying (dis)similarities of the data and in fact most of
the algorithms designed by previous work do not output ultrametrics
at all.
This paper takes a different perspective on the problem of computing
a hierarchical clustering: we are interested in how well the underlying
metric is preserved by the hierarchical clustering.
Also it is worth mentioning that in \cite{ABFPT99,AC11} the authors study various tree embedding with a focus on  average distortion in \cite{AC11}, and  tree metrics (and not ultrametrics) in \cite{ABFPT99}. 

Finally, a related but orthogonal approach to ours was taken in
recent papers by
\cite{cochez2015twister} and \cite{AbboudCH19}.
There, the authors design implementation of average-linkage and Ward's
method that have subquadratic running time by approximating the
greedy steps done by the algorithms. However, their results do not
provide any approximation guarantees in terms of any objective function
but rather on the quality of the approximation of the greedy step
and is not
guaranteed to produce an ultrametric.

\subsection{Organization of Paper}
This paper is organized as follows. In Section~\ref{sec:prelims} we introduce the Farach et al. algorithm. In Section~\ref{sec:approx} we introduce our near linear time approximation algorithm for general metrics, and in Section~\ref{sec:euclid} discuss its realization specifically in the Euclidean metric. In Section~\ref{sec:lower} we prove our conditional lower bounds on fast approximation algorithms. Finally, in Section~\ref{sec:expt} we detail the empirical performance of our proposed approximation algorithm.

\section{Preliminaries}
\label{sec:prelims}

Formally, an
  ultrametric $(X,\Delta)$ is a metric space where for each $x,y,z \in X$,
  $$\Delta(x,y) \le \max(\Delta(x,z),\Delta(z,y)).$$ For all finite point-sets $X$, it can be (always) realized in the following way as well. 
Let $T=(V,E)$ be a finite, rooted tree, and let $L$ denote the leaves of $T$. Suppose $w:V\setminus L\to\mathbb{R}^+$ is a function that assigns positive weights to the internal vertices of $T$ such that the vertex weights are non-increasing along root-leaf paths. Then one can define a distance on $L$ by
$$d_w(\ell,\ell'):=w(\LCA(\ell,\ell')),$$
where $\LCA$ is the least common ancestor. 
This is an ultrametric on $L$.

We consider the \textsc{Best Ultrametric Fit} problem (\BUF), namely:
\begin{itemize}
\item Input: a set $V$ of $n$ elements $v_1,\ldots,v_n$ and a weight
  function $w : V \times V \mapsto \R$.
  
\item Output: an ultrametric $(V,\Delta)$ such that $\forall v_i,v_j \in V$,  $w(v_i,v_j) \leq \Delta(v_i,v_j) \leq
  \alpha\cdot w(v_i,v_j)$, for the minimal value $\alpha$.
\end{itemize}

Note that we will abuse notation slightly and, for an edge
$e=(v_i,v_j)$, write $w(e)$ to denote $w(v_i,v_j)$. We write
$\dist^{\opt}$ to denote an optimal ultrametric, and let
$\alpha_{\opt}$ denote the minimum $\alpha$ for which
$\forall v_i,v_j \in V$,
$w(v_i,v_j) \leq \dist^{\opt}(v_i,v_j) \leq \alpha\cdot w(v_i,v_j)$.

We say that an ultrametric $\widehat{\dist}$ is a \emph{$\gamma$-approximation} to \BUF if
$\forall v_i,v_j \in V$,  $w(v_i,v_j) \leq \widehat{\dist}(v_i,v_j) \leq
\gamma \cdot \alpha_{\opt}\cdot w(v_i,v_j)$.
\vc{Alternatively, could be defined as
  $w(v_i,v_j) \leq \widehat{\dist}(v_i,v_j) \leq
  \gamma \cdot \dist^{\opt}(v_i,v_j)$. Not sure what is easier to prove.}
\gl{seems fine the way you defined it}
\cs{What about $1/\gamma\cdot w(v_i,v_j) \leq \widehat{\dist}(v_i,v_j) \leq
  \gamma \cdot \dist^{\opt}(v_i,v_j)$? Are these two equivalent? The lower bounds I prove hold for this relaxed notion as well. }

\gl{I don't know if this is equivalent but stating the way it is stated for now is maybe slightly better because it's the way Farach states it so we can use their results as a blackbox (I don't have a strong opinion on this, though)}

\vc{Ok, so maybe let's stick with this so that we don't have to change
anything in the rest of the paper today.}

\subsection{Farach-Kannan-Warnow's Algorithm}
% \vc{Not sure we really need this section since it seems that most of the stuff should be
%   re-proven so as to meet our approx def}

Farach et al.\ \cite{FKW95} provide an $O(n^2)$ algorithm to solve a
``more general'' problem (\ie that is such that an optimal algorithm for this problem can
  be used to solve \BUF),
the so-called ``sandwich
problem''. In the sandwich problem, the input consists of set $V$ of $n$ elements
$v_1,\ldots,v_n$ and two weight functions $w_{\ell}$ and $w_h$, and
the goal is to output an ultrametric $(V,\dist)$ such that
$\forall v_i,v_j \in V$,
$w_{\ell}(v_i,v_j) \leq \dist(v_i,v_j) \leq \alpha \cdot
w_h(v_i,v_j)$ for the minimal $\alpha$. Observe that an algorithm that
solves the sandwich problem can be used to solve \textsc{Best
  Ultrametric Fit} by setting $w_{\ell} = w_h = w$.

% , if such an ultrametric exists or report that the problem is
% infeasible otherwise.

We now review the algorithm of~\cite{FKW95}.
%% The running time of their algorithm is $O(n^2)$. % If we want an exact
% solution to our problem, we can reduce it to the sandwich problem by
% taking $G_l = d$ and $G_h = \alpha \cdot d$. Finding the optimal $\alpha$ is
% then obtained with a binary search.
Given a tree $T$ over the elements of $V$ and an
edge $e \in T$, removing $e$ from  $T$ creates two connected
components, we call $L(e)$ and $R(e)$ the set of elements in these connected
components respectively. Given $L(e)$ and $R(e)$, we define $P(e)$ to be the
set of pairs of elements $v_i \in L(e)$ and $v_j \in R(e)$ such that
the maximum weight of an edge of the path from $v_i$ to $v_j$ in $T$ is $w_{\ell}(e)$.

A \textbf{cartesian tree} of a weighted tree $T$ is a rooted tree
$T_C$ defined as follows: the root of $T_C$ corresponds to the edge of
maximal weight and the two children of $T_C$ are defined recursively
as the cartesian trees of $L(e)$ and $R(e)$, respectively. The leaves
of $T_C$ correspond to the nodes of $T$. Each node has an associated
height. The height of any leaf is set to $0$. For a non-leaf node
$u \in T_C$, we know that $u$ corresponds, by construction, to an edge
$e_u$ in $T$, which is the first edge (taken in decreasing order w.r.t.\
their weights) that separates $v_i$ from $v_j$ in $T$. Set the height
of $u$ to be equal to the weight of $e_u$ in $T$. A cartesian tree
$T_C$ naturally induces an ultrametric $\dist$ on its leaves: the distance
between two points $v_i$ and $v_j$ (\ie two leaves of $T_C$) is defined as
the height of their least common ancestor in $T_C$.

% for two points $v_i$ and $v_j$
% (\ie two leaves of $T_C$), let $u$ be the least common ancestor of
% $v_i$ and $v_j$ in $T_C$ and set $\dist^{T_C}(v_i,v_j)$ to be the height of $u$.
% where
% $\dist^{T_C}(v_i,v_j)$ equals to the weight of the edge corresponding
% to the least common ancestor of $v_i$ and $v_j$ in $T_C$.

Finally, we define the
\textbf{cut weight} of edge $e$ to be
$$\cw(e) = \max_{(v_i,v_j) \in P(e)} w_{\ell}(v_i,v_j).$$

The algorithm of~\cite{FKW95} is as follow:
\begin{enumerate}
\item Compute a minimum spanning tree (MST) $T$ over the complete graph $G_h$ defined on $V$ and with
  edge weights $w_h$;
\item Compute the cut weights with respect to the tree $T$;
\item Construct the cartesian tree $T_C$ of the tree $T'$ whose structure is identical to $T$ and the
  distance from an internal node of $T_C$ to the leaves of its subtree is given by the cut weight
  of the corresponding edge in $T$.
\item Output the ultrametric induced by the tree metric of $T_C$.
  %% \vc{Not so clear from the text right now
  %%   how the ultrametric $\dist^{T_C}$ comes from the tree
  %%   $T_C$.}\gl{is it better now? I defined the height of a node in the cartesian tree, simply say the distance between
  %% two points is the height of the least common ancestor. In the text above.}
\end{enumerate}

The following theorem is proved in~\cite{FKW95}:
\begin{theorem}
  \label{thm:farach}
  Given two weight functions $w_{\ell}$ and $w_h$, the above algorithm outputs an ultrametric $\dist$ such that for all $v_i, v_j \in V$
  $$w_{\ell}(v_i,v_j) \leq \dist(v_i,v_j) \leq \alpha_{\opt} \cdot w_h(v_i,v_j)$$
  for the minimal $\alpha_{\opt}$.
\end{theorem}

\section{\textsc{ApproxULT}: An Approximation Algorithm for \BUF}\label{sec:approx}
In this section, we describe a new approximation
algorithm for \BUF and prove its correctness.
We then show in the next section how it can be implemented efficiently
for inputs in the Euclidean metric.

Given a spanning tree $T$ over a graph $G$, any edge
$e=(v_i,v_j) \in G \setminus T$ induces a unique cycle $C^T_e$ which consists of the
union of $e$ and the unique path from $x$ to $y$ in $T$. We say that a
tree $T$ is a \emph{$\gamma$-approximate Kruskal tree} (or shortly a $\gamma$-KT) if
$$\forall e \in G\setminus T, w(e) \geq \frac{1}{\gamma} \max_{e' \in C^T_e}{w(e')}. $$

Moreover, given a tree $T$ and and an edge $e$ of $T$, we say that
$\beta \in \RR$ is a $\gamma$-estimate of $\cw(e)$ if
$\cw(e) \le \beta \le \gamma \cdot \cw(e)$. By extension, we say that a
function
$$\acw: V \times V \mapsto \R$$
is a $\gamma$-estimate of the cut weights $\cw$ if, for any edge $e$, $\acw(e)$ is a
$\gamma$-estimate of $\cw(e)$.

The rest of this section is dedicated to proving that the following
algorithm
achieves a $\gamma \delta$-approximation to \BUF,
for some parameters $\gamma \ge 1,\delta \ge 1$ of the algorithm.
\vc{parameters of alg or one of them is an absolute constant?}
\gl{Done. Yes, these are parameters of the algo. We can choose any combination of $\gamma \ge 1$ and $\delta\ge 1$ that we want. But for our Implementation (section 3), we instantiate $\delta$ to be $5$. I will write something in the corresponding section.}

\begin{enumerate}
\item Compute a $\gamma$-KT $T$ over the complete graph $G_h$ defined on
  $V$ and with
  edge weights $w_h$;
\item Compute a $\delta$-estimate $\acw$ of the cut weights of
  all the edge of the tree $T$;
\item Construct the cartesian tree $T_C$ of the tree $T'$ whose structure is identical to $T$ and the
  distance from an internal node of $T_C$ to the leaves of its subtree is given by the $\acw$
  of the corresponding edge in $T$.
\item Output the ultrametric $\dist$ over the leaves of $T_C$.
\end{enumerate}

We want to prove the following:

\begin{theorem}
  \label{thm:main}
    For any $\gamma \ge 1, \delta \ge 1$, the above algorithm outputs an ultrametric $\dist$ which is a
  $\gamma \delta$-approximation to \BUF, meaning
  that for all $v_i, v_j \in V$
  $$w_{\ell}(v_i,v_j) \leq \dist(v_i,v_j) \leq \gamma \cdot \delta \cdot \alpha_{\opt} \cdot w_h(v_i,v_j)$$
\end{theorem}

% \vc{I think that this subsection should contain the whole proof, almost copy-pasted from~\cite{FKW95}
%   but that shows that approximating Steps 1 and 2 is ok. See further comments below.}

\begin{proof}
  $ $
  
  \noindent \textbf{First step: }we prove that the $\gamma$-KT $T$ computed at
  the first step of the algorithm can be seen an exact MST for a
  complete weighted graph $G'$ defined on $V$ and with a weight
  function $w'$ satisfying

  $$ \forall v_i, v_j \in V, w'(v_i, v_j) \leq \gamma \cdot w_h(v_i, v_j).$$

  We construct $w'$ in the following way. For each pair of points $(v_i,v_j)$:
  \begin{itemize}
  \item If $(v_i,v_j) \in T$, then set $w'(v_i,v_j) = w_h(v_i,v_j)$
  \item If $(v_i,v_j) \not \in T$, then set $w'(v_i,v_j) = \gamma w_h(v_i,v_j)$.
  \end{itemize}

  By construction, it is clear that $w' \leq \gamma \cdot w_h$. To see
  that $T$ is an (exact) MST of $G'$, consider any MST $F$ of $G'$. If
  $e = (v_i,v_j) \in F \setminus T$, then consider the first edge $e'$
  in the unique path from $v_i$ to $v_j$ in $T$ that reconnects
  $F\setminus e$. By definition of $w'$, we have $w'(e') = w_h(e')$
  and $w'(e) = \gamma w_h(e)$. Since $T$ is a $\gamma$-KT, we also
  have that $w_h(e) \geq \frac{1}{\gamma} w_h(e')$. Therefore
  $w'(e') \leq w'(e)$ and $\{F \cup e'\} \setminus e$ is a spanning
  tree of $G'$ of weight smaller than or equal to the weight of
  $F$. This proves that $\{F \cup e'\} \setminus e$ is also a
  MST. Doing this process for all edges not in $T$ gives eventually
  $T$ and proves that $T$ is also a MST of $G'$, as desired.

  \noindent \textbf{Second step.} Observe that the weight function $w_h$ is
  not involved in steps 2, 3, and 4 of the algorithm. Therefore, if
  steps 2, 3, and 4 of the algorithm were made without approximation (meaning
  that we compute the exact cut weights $\cw$ associated to the
  $\gamma$-KT tree $T$ and we output the ultrametric to the
  corresponding cartesian tree), then the output would be an ultrametric
  $\dist$ such that for all $v_i, v_j \in V$
  \begin{equation}
    \label{eq:1}
    w_{\ell}(v_i, v_j) \leq \dist(v_i, v_j) \leq \alpha'_{\opt} \cdot w'(v_i, v_j)
  \end{equation}
  for the minimal such $\alpha'_{\opt}.$ This follows directly from
  Theorem~\ref{thm:farach} and the fact that $T$ is an exact MST for
  the graph $G'$ defined above. Note that
  $\alpha'_{\opt} \leq \alpha_{\opt}$ where $\alpha_{\opt}$ denotes
  the minimal constant such that there exists an ultrametric between
  $w_l$ and $\alpha_{\opt} \cdot w_h$.

  Now, consider the ultrametric $\dist^{T_C}$ associated to $T$ and a $\delta$-estimate $\acw$ of the cut weights. We claim that for all $v_i, v_j \in V$
  \begin{equation}
    \label{eq:2}
    \dist(v_i, v_j) \leq \dist^{T_C}(v_i, v_j) \leq \delta \cdot \dist(v_i, v_j).
  \end{equation}

  To see this, take any $v_i, v_j \in V$.  By definition,
  $\dist^{T_C}(v_i,v_j) = \acw(e)$ for the first edge $e$ (taken in
  decreasing order w.r.t.\ to $\acw$) that separates $v_i$ from $v_j$ in
  $T$. Let $e_{v_i,v_j}$ be the first edge that separates $v_i$ from
  $v_j$ w.r.t.\ to the actual cut weights $\cw$. Again, we have by
  definition that $\dist(v_i, v_j) = \cw(e)$. We have that
  $\acw(e) \geq \acw(e_{v_i,v_j})$ since $e$ is the first edge w.r.t.\
  $\acw$. Moreover $\acw(e_{v_i, v_j}) \geq \cw(e_{v_i, v_j})$ because
  $\acw$ is a $\delta$-estimate of the cut weights: this gives us the
  first desired inequality
  $$\dist^{T_C}(v_i, v_j) \geq \dist(v_i, v_j).$$

  The upper bound is similar. We know that
  $\acw(e) \leq \delta \cdot \cw(e)$ since $\acw$ is a $\delta$-estimate. We
  also have that $\cw(e) \leq \cw(e_{v_i, v_j})$ since $e_{v_i, v_j}$
  is the first separating edge w.r.t.\ $\cw$. This gives:
  $$\dist^{T_C}(v_i, v_j) \leq \delta \cdot \dist(v_i, v_j).$$

  All together, Equations~\ref{eq:1} and~\ref{eq:2} imply

  \begin{align*}
    w_{\ell}(v_i, v_j) \leq \dist^{T_C}(v_i, v_j) &\leq \gamma \cdot \alpha'_{\opt} \cdot w'(v_i, v_j)\\
                                                  &\leq \gamma \cdot \delta \cdot \alpha'_{\opt} \cdot w_h(v_i, v_j) \\
                                                  &\leq \gamma \cdot \delta \cdot \alpha_{\opt} \cdot w_h(v_i, v_j)
  \end{align*}
  as desired.
\end{proof}

\section{A Fast Implementation of \textsc{ApproxULT} in Euclidean Space
-- Proof of Theorem~\ref{thm:UB:euclid}}\label{sec:euclid}
In this section, we consider inputs of \BUF that consists
of a set of points $V$ in $\RR^d$, and so for which
$w(v_1,v_2) = \|v_1 - v_2\|_2$.
%We show the following theorem.
%\vc{Refer to this theorem in the intro}
We now explain how to implement \textsc{ApproxULT} efficiently for
$\gamma \ge 1$ and $\delta = 5$.\gl{Here, we have any $\gamma$ but $\delta = 5$}

\paragraph{Fast Euclidean $\gamma$-KT.}
For computing efficiently a $\gamma$-KT of a set of points in a Euclidean
space of dimension $d$, we appeal to the result of
\cite{har2013euclidean}
(if interested in doubling metrics,
one can instead use the bound of \cite{FiltserN18}).
The approach relies on \emph{spanners}; A $c$-spanner 
of a set $S$ of $n$ points in $\RR^d$ is a graph $G = (S,E)$ and a
weight function $w : E \mapsto \RR_+$ such that for any $u,v\in S$,
the shortest path distance in $G$ under the edge weights induced by $w$,
$\dist^G(u,v)$ satisfies $\|u-v\|_2 \le \dist^G(u,v) \le
c \cdot \|u-v\|_2$.

The result of \cite{har2013euclidean} states that there is
an algorithm that for any set $S$ of
 $n$ points in $\RR^d$ produces 
an $O(\gamma)$-spanner for $S$ with $O(n^{1+1/c^2}\log^2 n)$ edges 
in time $O(nd + n^{1+1/c^2}\log^2 n)$.
The algorithm uses the 
locality sensitive hash family of \cite{andoni2006near}, or alternatively
for $\gamma = \sqrt{\log n}$ the  Lipschitz partitions
of \cite{charikar1998approximating}. 

An immediate application of Kruskal 
classic algorithm for computing
a minimum spanning tree on the spanner yields an algorithm
with running time $O(nd + n^{1+1/c^2}\log^3 n)$. Moreover, we claim
that a minimum spanning tree on a $c$-spanner $G$ is indeed a $c$-KT
for the original point set. Assume towards contradiction that this
is not the case.
Then there exists an edge $e = (u,v) \not\in T$ such that
$\|u-v\|_2 < \max_{(x,y) \in C_e^T} \|x-y\|_2/c$. By correctness of the
$c$-spanner we have that $\dist^G(u,v) \le c\|u-v\|_2 <
\max_{(x,y) \in C_e^T} \|x-y\|_2 \le \max_{(x,y) \in C_e^T} \dist^G(x,y)$.
A contradiction to the fact that $T$ is an MST of the $c$-spanner.

%% \subsection{Computing a $\gamma$-MST}
%% \gl{todo}
\paragraph{Fast Estimation of the Cut Weights.}
\gl{Fixed the complexity, it should be $O(nd + n \log n)$ I guess}We explain how to compute in time $O(nd + n \log n)$ a $5$-estimate of the cut weights. To do this, we maintain a disjoint-set data structure on
$X$ with the additional property that each equivalence class $C$ (we
call such an equivalence class \emph{cluster}) has a special vertex
$r_C$ and we store $m_C$ the maximal distance between $r_C$ and a
point in the cluster. We now consider the edges of the MST $T$ in
increasing order (w.r.t.\ their weights). When at edge $e = (x,y)$, we look at the two
clusters $C$ and $D$ coming from the equivalence classes that
respectively contain $x$ and $y$. We claim that
$$E = 5\cdot \max (d(r_C,r_D), m_C-d(r_C,r_D), m_D-d(r_C,r_D)) $$
is a $5$-approximation of the cut weight for $e$. To see this, observe that if $x',y'$ are the farthest points respectively in $C, D$, then:\allowdisplaybreaks
\begin{align*}
  d(x',y') &\leq d(x',r_C) + d(r_C,r_D) + d(r_D,y') \\
         &\leq d(x',r_C)-d(r_C,r_D) + 3d(r_C,r_D) + d(r_D,y')-d(r_C,r_D)\\
         &\leq 5.\max (d(r_C,r_D), m_C-d(r_C,r_D), m_D-d(r_C,r_D))\leq E
\end{align*}
On the other hand\allowdisplaybreaks
\begin{align*}
  d(r_C,r_D) & \leq d(x',y')\\
  m_C - d(r_C,r_D) & \leq d(x',r_D) \leq d(x',y') \\  
  m_D - d(r_C,r_D) & \leq d(y',r_C) \leq d(x',y') 
\end{align*}
and therefore $E \leq 5\cdot d(x',y')$. Finally,  if we consider the path from $x'$ to $y'$ in $T$, it is clear that the pair $(x',y')$ is in $P(e)$, and the bound on $CW(e)$ follows. 

Merging $C$ and $D$ can simply be done via a classic disjoint-set data
structure. Thus, the challenge is to update $m_{C \cup D}$. To do so,
we consider the smallest cluster, say $D$, query $d(x,r_C)$ for each
point $x \in D$ and update accordingly $r_{C\cup D}$ if a bigger value
is found. Therefore the running time to update $m_{C \cup D}$ is
$O(|D|\times d)$ (we compute $|D|$ distances in a space of dimension
$d$). The overall running time to compute the approximate cut weights
is $O(nd + n \log n)$: sorting the edges requires $O(n \log n)$ and
constructing bottom-up the cut-weights with the disjoint-set data
structure takes $O(nd + n \alpha(n))$, where $\alpha(n)$ denotes the
inverse of the Ackermann function (this part comes from the disjoint-set
structure). To conclude, note that $n \alpha(n)$ is much smaller than
$n \log n$.

% The running time to merge the two clusters is $O(|D|\times d + n \alpha(n))$ (we compute $|D|$ distances in a space of dimension $d$), so
% the overall running time to compute the approximate cut weights is
% $O(nd + n \log n)$ (sorting the edges requires $O(n \log n)$ and constructing
% bottom-up the cut-weights with the disjoint-set data structure takes $O(n \alpha(n))$).

\section{Hardness of \BUF for High-Dimensional  
    Inputs}\label{sec:lower}
We complement Theorem~\ref{thm:UB:euclid}
with a  hardness of approximation result in this section.
Our lower bound is based on the well-studied Strong Exponential Time Hypothesis (\SETH) \cite{IP01,IPZ01,CIP06} which roughly states that SAT on $n$ variables cannot be solved in time less than $2^{n(1-o(1))}$. \SETH is a popular assumption to prove lower bounds for problems in \P\ (see the following surveys \cite{Williams15,Williams16,Vir18,RW19} for a discussion).

\begin{theorem}\label{thm:ellinf}
  Assuming \SETH, for every $\varepsilon>0$,  no algorithm running in time $n^{2-\varepsilon}$ can, given as input an instance of  \BUF consisting $n$ points  of dimension $d:=O_{\varepsilon}(\log n)$ in $\ell_\infty$-metric, distinguish between the following two cases.
  \begin{description}
  \item[Completeness:] There is an isometric ultrametric embedding. 
  \item[Soundness:] The distortion of the best ultrametric embedding is at least $\nicefrac{3}{2}$.
  \end{description}
\end{theorem}

Note that the above theorem morally\footnote{We say ``morally'' because our hardness results are for the decision version, but doesn't immediately rule out algorithms that find approximately optimal embedding, as computing the distortion of an embedding (naively) requires $n^2 $ time. So the search variant cannot be naively reduced to the decision variant.} rules out approximation algorithms running in subquadratic time which can approximate the best ultrametric to $\nicefrac{3}{2}-o(1)$ factor. 

Finally, we  remark that all the results in this section can be based on a weaker assumption called the Orthogonal Vectors Hypothesis \cite{W05} instead of \SETH.
Before we proceed to the proof of the above theorem, we prove below a key technical lemma. 

\begin{definition}[Point-set $S^*$]
For every $\gamma,\gamma'\ge 0$ and every $p\in \mathbb{R}_{\ge 1}\cup \{\infty\}$, we define the discrete point-set $S^*(\gamma,\gamma',p):=\{a,a',b\}$ in the $\ell_p$-metric as follows: $$ \|a-b\|_p\le 1,\ \|a-a'\|_p\le 1+\gamma',  \text{ and }\|a'-b\|_p\ge 1+\gamma.$$
\end{definition}

\begin{lemma}[Distortion in Ultrametric Embedding]\label{lem:distortion}
Fix $\gamma,\gamma' \ge 0$ and $p\in \mathbb{R}_{\ge 1}\cup \{\infty\}$. Then we have that any embedding of $S^*(\gamma,\gamma', p):=\{a,a',b\}$ into ultrametric incurs a distortion of at least $\frac{1+\gamma}{1+\gamma'}$.
\end{lemma}
\begin{proof} 
Let the distortion of $S^*$ to the ultrametric be at most $\rho$. Let $\tau$ be the embedding into ultrametric with distortion $\rho$ and let  $\Delta$ denote distance in the ultrametric. Let $\alpha\in\mathbb{R}^+$ be the scaling factor of the embedding from the $\ell_p$-metric to the ultrametric.
\begin{align*}
(1+\gamma)\cdot \alpha&\le \Delta(\tau({a'}),\tau({b}))\\
&\le\max\{\Delta(\tau(a),\tau(b)),\Delta(\tau(a),\tau({a'}))\}\\
&\le \rho\cdot (1+\gamma')\cdot \alpha
\end{align*}
Thus we have that $\rho\ge \frac{1+\gamma}{1+\gamma'}$.
\end{proof}

We combine the above lemma with David et al.'s conditional lower bound (stated below) on approximating the Bichromatic Closest Pair problem in the $\ell_\infty$-metric to obtain Theorem~\ref{thm:ellinf}.

\begin{theorem}[\cite{DKL19}]\label{thm:DKL19}
Assuming \SETH, for any $\varepsilon>0$, no algorithm running in time $n^{2-\varepsilon}$,
given  $A,B\subseteq \mathbb{R}^d$ as input, where $|A|=|B|=n$ and $d=O_\varepsilon(\log n)$, distinguish between the following two cases:
\begin{description}
\item[Completeness:] There exists $(a,b)\in A\times B$ such that $\|a-b\|_\infty=1$.
\item[Soundness:] For every $(a,b)\in A\times B$ we have $\|a-b\|_\infty=3$.
\end{description}
Moreover this hardness holds even with the following additional properties:
\begin{itemize}
\item Every distinct pair of points in $A$ (resp.\ $B$) are at distance 2 from each other in the $\ell_\infty$-metric.
\item All pairs of points in $A\times B$ are at distance either 1 or 3 from each other in the $\ell_\infty$-metric.
\end{itemize}
\end{theorem}

\begin{proof}[Proof of Theorem~\ref{thm:ellinf}]
Let $(A,B)$ be the input to the hard instances of the Bichromatic Closest Pair problem as given in the statement of Theorem~\ref{thm:DKL19} (where $A,B\subseteq \mathbb{R}^d$ and $|A|=|B|=n$). We show that if for every $(a,b)\in A\times B$ we have $\|a-b\|_\infty=3$ then there is an isometric embedding of $A\cup B$ into an ultrametric and if there exists $(a,b)\in A\times B$ such that $\|a-b\|_\infty=1$ then any embedding of $A\cup B$ to an ultrametric incurs a distortion of $\nicefrac{3}{2}$. Once we show this, the proof of the theorem statement immediately follows.

Suppose that for every $(a,b)\in A\times B$ we have $\|a-b\|_\infty=3$. We construct the following ultrametric embedding. Let $T$ be a tree with root $r$. Let $r$ have two children $c_A$ and $c_B$. Both $c_A$ and $c_B$ each have $n$ leaves which we identify with the points in $A$ and points in $B$ respectively. Then we subdivide the edge between $c_A$ and its leaves and $c_B$ and its leaves.  Notice that any pair of leaves corresponding to two distinct points in $A$ (resp.\ in $B$) are at distance four away in $T$. Also notice that any pair of leaves corresponding to a pair of points in $A\times B$ are at distance six. Therefore the aforementioned embedding is isometric. 

Next, suppose that there exists $(a,b)\in A\times B$ such that $\|a-b\|_\infty=1$. We also suppose that there exists $(a',b)\in A\times B$ such that $\|a'-b\|_\infty=3$. We call Lemma~\ref{lem:distortion} with the point-set $\{a,a',b\}$ and parameters $\gamma=2$ and $\gamma'=1$. Thus we have that even just embedding $\{a,a',b\}$ into an ultrametric incurs distortion of $\nicefrac{3}{2}$.
\end{proof}

One may wonder if one can extend Theorem~\ref{thm:ellinf} to the Euclidean metric to rule out approximation algorithms running in subquadratic time which can approximate the best ultrametric to arbitrary factors close to 1. More concretely, one may look at the hardness of approximation results of \cite{R18,KM19} on Closest Pair problem, and try to use them as the starting point of the reduction. An immediate obstacle to do so is that in the soundness case of the closest pair problem (i.e., the completeness case of the computing ultrametric distortion problem), there is no good bound on the range of all pairwise distances, and thus the distortion cannot be estimated to yield a meaningful reduction.

Nonetheless, we introduce a new complexity theoretic hypothesis below and
show how that extends Theorem~\ref{thm:ellinf} to the Euclidean metric.

\paragraph{Colinearity Hypothesis.}
Let \B denote the $d$-dimensional unit Euclidean ball.  In the Colinearity Problem (\CP), we are given as input a set $A$ of $n$ vectors uniformly and independently sampled from \B, and we move one of these sampled points to be \emph{closer} to the midpoint of two other sampled points. The goal is
to find these three points. More formally, we can write it as a decision problem in the following way. 

Let \Dun{}$(n,d)$ be the distribution which samples $n$ points uniformly and independently from \B. For every $\rho\in [0,1]$, let \Dpl{}$(n,d,\rho)$ be the following distribution: 
\begin{enumerate}
\item Sample $(a_1,\ldots ,a_n)\sim $\Dun{}$(n,d)$.
\item Pick three distinct indices $i,j,k$ in $[n]$ at random.  
\item Let $a_{i,j}$ be the midpoint of $a_i$ and $a_j$. 
\item Let $\tilde a_k$ be $(1-\rho)\cdot a_k + \rho\cdot a_{i,j}$.
\item Output $(a_1,\ldots ,a_{k-1},\tilde a_k,a_{k+1},\ldots ,a_n)$.
\end{enumerate}

Notice that \Dun{}$(n,d)$ $=$\Dpl{}$(n,d,0)$. Also, notice that in \Dpl{}$(n,d,1)$ we have planted a set of three colinear points. The decision problem \CP would then be phrased as follows.

\begin{definition}[\CP]
Let $\rho \in (0,1]$. Given as input a set of $n$ points sampled from \Dun{}$(n,d)$ $\cup$ \Dpl{}$(n,d,\rho)$, distinguish if it was sampled from \Dun{}$(n,d)$ or from \Dpl{}$(n,d,\rho)$.
\end{definition}

The worst case variant of \CP has been studied extensively in computational geometry and more recently in fine-grained complexity. In the worst case variant, we are given a set of $n$ points in $\mathbb{R}^d$ and we would like to determine if there are three points in the set that are colinear. This problem can be solved in time $O(n^2 d)$. It's now known  that this runtime cannot be significantly improved assuming the \textsf{3-SUM} hypothesis \cite{GO95,GO12}. We putforth the following hypothesis on \CP:

\begin{definition}[Colinearity Hypothesis (\CH)]
There exists constants $\rho,\varepsilon>0$ such that no randomized algorithm running in time $n^{1+\varepsilon}$ can decide \CP (with parameters $n,d,\rho$), for every $d\ge O_{\rho,\varepsilon}(\log n)$.  
\end{definition}

Notice that unlike \OVH or \textsf{3-SUM} hypothesis, we are \emph{not} assuming a subquadratic hardness for \CP, but only assume a superlinear hardness, as \CP is closely related to the \emph{Light bulb problem} \cite{V88}, for which we do have subquadratic algorithms \cite{V15,KKK16,A19}. Elaborating, we now provide an informal sketch of a reduction from \CP to the Light bulb problem: given $n$ points sampled from \Dun{}$(n,d)$ $\cup$ \Dpl{}$(n,d,\rho)$, we first apply the sign function (+1 if the value is positive and -1 otherwise) to each coordinate of the sampled points, to obtain points on the Boolean hypercube. Then we only retain each point w.p.\ $\nicefrac{1}{2}$ and discard the rest.  If the points were initially sampled from \Dun{}$(n,d)$ then the finally retained points will look like points sampled uniformly and independently from the Boolean hypercube, whereas,  if the points were initially sampled from \Dpl{}$(n,d,\rho)$ then there are two pairs of points that are $\rho'$-correlated ($\rho'$ depends on $\rho$) after applying the sign function and exactly one of the two pairs is retained with constant probability. 

Returning to the application of \CH to ultrametric embedding, assuming \CH, we  prove the following result.

\begin{theorem}\label{thm:ellp}
  Assuming \CH,  there exists $\varepsilon,\delta>0$ such that no randomized algorithm running in time $n^{1+\varepsilon}$ can given as input an instance of  \BUF consisting of $n$ points  of dimension $d:=O_{\varepsilon,\delta}(\log n)$ in Euclidean metric distinguish between the following two cases. 
  \begin{description}
  \item[Completeness:] The distortion of the best ultrametric embedding is at most $1+\nicefrac{\delta}{2}$. 
  \item[Soundness:] The distortion of the best ultrametric embedding is at least $1+\delta$.
  \end{description}
\end{theorem}

We use the following fact about  random sampling from high-dimensional unit ball.

\begin{fact}[\cite{Vbook18}]\label{fact}
\begin{sloppypar}For every $\delta>0$ there exists $c\in\mathbb N$ such that the following holds. Let $(a_1,\ldots ,a_n)\sim $\Dun{}$(n,c\cdot \log n)$. Then with high probability we have that for all distinct $i,j$ in $[n]$, $$\|a_i-a_j\|_2\in (\beta-\delta,\beta+\delta),$$ for some universal scaling constant $\beta>1$. \end{sloppypar}
\end{fact}

\begin{proof}[Proof of Theorem~\ref{thm:ellp}]
Let $\varepsilon,\rho$ be the constants from \CH. Let $\delta:=\rho/9$ and $c$ be an integer guaranteed from Fact~\ref{fact}.
Let $A$ be the input to \CP (where $A\subseteq $\B and $|A|=n$). We may assume that $d>c\log n$. We show that if all points in $A$ were picked  independently and uniformly at random from \B then there is an  embedding of $A$ into an ultrametric with distortion less than $1+2\delta$ and if otherwise $A$ was sampled from \Dpl{}$(n,d,\gamma)$ then any embedding of $A$ to an ultrametric incurs a distortion of $1+4\delta$. Once we show this, the proof of the theorem statement immediately follows.

Suppose that  $A$ was sampled from \Dun{}$(n,d)$. From Fact~\ref{fact} we have that for all distinct $a_i,a_j$ in $A$, $$\|a_i-a_j\|_2\in (\beta-\delta,\beta+\delta),$$ for some universal scaling constant $\beta>1$.  Then the  ultrametric embedding is simply given by identifying $A$ with the leaves of a star graph on $n+1$ nodes. The distortion in the embedding in such a case would be at most $\frac{\beta+\delta}{\beta-\delta}\le 1+\nicefrac{2\delta}{\beta}<1+2\delta$.

Next, suppose that $A$ was sampled from \Dpl{}$(n,d,\rho)$. Then there exists 3 points $a_i,a_j,\tilde a_k$ in $A$ such that the following distances hold: 
\begin{align*}
\|a_i-a_j\|_2&\ge \beta-\delta,\\
\|a_i-\tilde a_k\|_2,\|a_j-\tilde a_k\|_2
&\le \sqrt{(\nicefrac{(\beta+\delta)}{2})^2+\nicefrac{3}{4}((1-\rho)\cdot(\beta+\delta))^2}
\le \beta -\rho.
\end{align*}

 We call Lemma~4.3 with the point-set $\{a_i,a_j,\tilde a_k\}$. Thus we have that even just embedding $\{a_i,a_j,\tilde a_k\}$ into an ultrametric incurs distortion of $1+4\delta$.
\end{proof}

Note that we can replace \CH by a search variant and this would imply the lower bound to the search variant of the \BUF problem (unlike Theorem~4.1). 

\section{Experiments}\label{sec:expt}
We present some experiments performed on three standard datasets:
DIABETES (768 samples, 8 features), MICE (1080 samples, 77 features),
PENDIGITS (10992 samples, 16 features) and compare our C++
implementation of the algorithm described above to the classic linkage
algorithms (average, complete, single or ward) as implemented in
the Scikit-learn library (note that the Scikit-learn implementation is also in C++).
The measure we are interested in is the maximum
distortion $\max\limits_{(u,v) \in P} \frac{\Delta(u,v)}{\|u-v \|_2}$,
where $P$ is the dataset and $\Delta$ the ultrametric output by 
the algorithm. Note that average linkage, single and ward linkage can
underestimate distances, \ie $\frac{\Delta(u,v)}{\|u-v \|_2} < 1$ for
some points $u$ and $v$. In practice, the smallest ratio given by
average linkage lies often between $0.4$ and $0.5$ and between $0.8$
and $0.9$ for ward linkage. For single linkage, the maximum distortion
is always $1$ and hence the minimum distortion can be very small. For
a fair comparison, we normalize the ultrametrics by multiplying every
distances by the smallest value for which
$\frac{\Delta(u,v)}{\|u-v \|_2}$ becomes greater than or equal to $1$
for all pairs. Note that what matters most in hierarchical clustering
is the structure of the tree induced by the ultrametric and performing
this normalization (a uniform scaling) does not change this structure.

\textbf{ApproxULT} \vc{change name, if we decide to stick to
  ApproxULT, then use ApproxULT. In any case, we should not use
  ALGO1.}\gl{for now, I changed it for ApproxULT. If we change again,
  we will need to change the name in the plot as well. Is ApproxULT2 a
  good name for the second algo?}  stands for the C++ implementation
of our algorithm. To compute the $\gamma$-approximate Kruskal tree, we
implemented the idea from~\cite{har2013euclidean}, that uses the
locality-sensitive hash family of \cite{andoni2006near} and runs in
time $O(nd + n^{1+1/\gamma^2}\log^2 n)$. The parameter $\gamma$ is
related to choices in the design of the locality-sensitive hash
family. It is hard to give the precise $\gamma$ that we choose during
our experiments since it relies on theoretical and asymptotic
analysis. However, we choose parameters to have, in theory, a $\gamma$
around $2.5$. %, as already said above.
\vc{We should say which value of $\gamma$ was chosen for the
  exps.}\gl{ok, I said something regarding this, hope that is
  understandable.}  Observe that our algorithm is roughly cut into two
distinct parts: computing a $\gamma$-KT tree $T$, and using $T$ to
compute the approximate cut weights and the corresponding cartesian
tree. Each of these parts play a crucial role in the approximation
guarantees. To understand better how important it is to have a tree
$T$ close to an exact MST, we implemented a slight variant of
\textbf{ApproxULT}, namely \textbf{ApproxAccULT}, in which $T$ is
replaced by an exact MST. Finally, we also made an implementation of
the quadratic running time \textbf{Farach et al.}'s algorithm since it
finds an optimal ultrametric. The best known algorithm for computing
an exact MST of a set of high-dimensional set of points is
$\Theta(n^2)$ and so \textbf{ApproxAccULT} and \textbf{Farach et al.'s
  algorithm} did not exhibit a competitive running time and were not
included in Figure~\ref{fig:runningTime}.

% Since the algorithm is roughly cut in two distinct parts
% (computing a $\gamma$-KT tree $T$, and using $T$ to compute the
% approximate cut weights and the corresponding cartesian tree), and
% each of these parts is crucial in the approximation guarantees, we
% also implemented a slight variant, \textbf{ALGO2}, in which $T$ is
% replaced by an exact MST. The purpose of \textbf{ALGO2} is to
% understand how important it is to have a tree $T$ close to a perfect
% MST. 

Table~\ref{fig:distortions} shows the maximum distortions of the
different algorithms. \textbf{Farach et al.} stands for the baseline
since the algorithm outputs the best ultrametric. For the linkage
algorithms, the results are deterministic hence exact (up to rounding)
while the output of our algorithm is probabilistic (this probabilistic
behavior comes from the locality-sensitive hash families). We
performed 100 runs for each dataset.\vc{we should state what is the
  variance}\gl{don't have time to compute it precisely. We can do this
  later. I slightly changed the text to avoid mentioning variance.}
 We observe that \textbf{ApproxULT} performs better than Ward's
method while being not too far from the others. \vc{give approx
  factors}\gl{is that $\gamma$? We already say something above,
  wouldn't it be redundant?}  \textbf{ApproxAccULT} performs almost
better than all algorithms except single linkage, this emphasizes the
fact that finding efficiently accurate $\gamma$-KT is
important. Interestingly single linkage is in fact close to the optimal solution.

\vc{compare with what the theory predicts for the $\gamma$ we
  chose}\gl{I am not sure what you have in mind. The complexity
  analysis and approx factor are asymptotic with big O so I don't know
  how to compare?}

\begin{table}
\begin{center}
\begin{tabular}{|c|c|c|c|}
  \hline
     & DIABETES & MICE & PENDIGITS \\
  \hline
  Average  & 11.1 & 9.7 & 27.5 \\
  Complete  & 18.5 & 11.8 & 33.8 \\
  Single  & 6.0 & 4.9 & 14.0 \\
  Ward  &  61.0 &59.3 & 433.8\\
  ApproxULT  & 41.0  & 51.2 & 109.8 \\
  ApproxAccULT  & 9.6 & 9.4 & 37.2 \\
  Farach et al.  & 6.0 & 4.9 & 13.9 \\
  \hline
\end{tabular}
\caption{Max distortions}
\label{fig:distortions}
\end{center}
\end{table}

Figure~\ref{fig:runningTime} shows the average running time, rounded
to $10^{-2}$ seconds. We see
that for small datasets, \textbf{ApproxULT} is comparable to linkage
algorithms, while \textbf{ApproxULT} is much faster on a large dataset, as
the complexity analysis predicts (roughly $36$ times faster than the
slowest linkage algorithm and $10$ times faster than the fastest one).
\vc{compare with what the theory predicts for the $\gamma$ chosen}\gl{same as above, I don't see how to do this}
\begin{figure}[h]
\begin{center}  \includegraphics[scale=0.6]{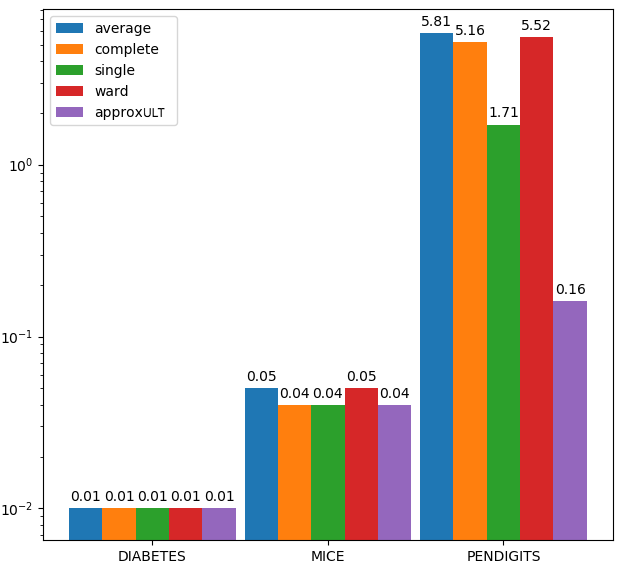}
\caption{Average running time, in seconds. Logarithmic scale.}
\label{fig:runningTime}
\end{center}
\end{figure}

\vc{THIS IS NOT A COMMENT -- but rather a todo. %when we have time.
  Before the rebuttal:
  \begin{itemize}
  \item Fix the LSH scheme;
  \item Plots experiments with different values of $\gamma$ to
    see how this changes the running time
  \end{itemize}

  For later:
  \begin{itemize}
  \item Try with tree embedding ($\sqrt{\log n}$ distortion).
  \item Compute the best distortion using Farach et al.    
  \end{itemize}
}
\gl{Yes, good points, we should do this at some point!}

\section*{Acknowledgements}
We would like to thank all the reviewers for various comments that improved the presentation of this paper. We would also like to thank Ronen Eldan and Ori Sberlo for discussions on concentration of Gaussian. 

Karthik C.\ S.\ would like to thank the support of the  Israel Science Foundation (grant number 552/16) and the Len Blavatnik and the Blavatnik Family foundation. Guillaume Lagarde would like to thank the support of the DeepSynth CNRS Momentum project. Ce projet a b\'en\'efici\'e d'une aide de l'\'Etat g\'er\'ee
    par l'Agence Nationale de la Recherche au titre du Programme
    Appel à projets générique JCJC 2018 portant la r\'ef\'erence
    suivante : ANR-18-CE40-0004-01.

\bibliographystyle{alpha}
\bibliography{references}

\end{document}

\bibliographystyle{alpha}
\bibliography{references}

\end{document}